\documentclass[journal]{IEEEtran}

\usepackage{amsmath}
\usepackage{amsfonts}
\usepackage{amssymb}
\usepackage{amsthm}
\usepackage{tabularx}
\usepackage{mathrsfs} 
\usepackage{mathtools}

\usepackage{url}
\usepackage{hyperref}

\newtheorem{Proposition}{Proposition}

\newtheorem{corollary}{Corollary}

\usepackage{stfloats}
\usepackage{float}
\usepackage{graphicx}
\usepackage{cite}
\usepackage{xcolor}
\usepackage{subfigure}

\makeatletter
\def\blfootnote{\xdef\@thefnmark{}\@footnotetext}
\makeatother

\begin{document}
	
	\title{\huge{Performance Analysis of Fluid Antenna-aided Backscatter Communications Systems}
	} 
	\author{Farshad~Rostami~Ghadi, \IEEEmembership{Member}, \textit{IEEE},~Masoud~Kaveh,~\IEEEmembership{Member}, \textit{IEEE},~Kai-Kit~Wong, \IEEEmembership{Fellow},~\textit{IEEE}
	}
	\maketitle
	\begin{abstract}
		This paper studies the performance of backscatter communications (BC) over emerging fluid antenna (FA) technology. In particular, a single-antenna source sends information to a FA reader through the wireless forward (i.e., source-to-tag) and backscatter (tag-to-reader) channels. For the considered BC, we first derive the cumulative distribution function (CDF) of the equivalent channel at the FA receiver, and then we obtain closed-form expressions of the outage probability (OP) and delay outage rate (DOR) under a correlated Rayleigh distribution. Moreover, in order to gain more insights into the system performance, we present analytical expressions of the OP and DOR at the high SNR regime. Numerical results indicate that considering the FA at the reader can significantly improve the performance of BC in terms of the OP and DOR compared with a single-antenna reader.
	\end{abstract}
	\maketitle
\blfootnote{This work is supported by the Engineering and Physical Sciences Research Council (EPSRC) under Grant EP/W026813/1. For the purpose of open access, the authors will apply a Creative Commons Attribution (CC BY) licence to any Author Accepted Manuscript version arising.}
\blfootnote{\noindent Farshad Rostami Ghadi and  Kai-Kit Wong are with the Department of Electronic and Electrical Engineering, University College London, WC1E 6BT London, UK. (e-mail:$\{\rm f.rostamighadi,kai\text{-}kit.wong\}@ucl.ac.uk$).}
\blfootnote{\noindent Masoud Kaveh is with the Department of Information and Communication Engineering, Aalto University, 02150 Espoo, Finland. (e-mail:$\rm masoud.kaveh@aalto.fi$) }

\blfootnote{Digital Object Identifier 10.1109/XXX.2021.XXXXXXX}
\vspace{0mm}
\section{Introduction}\label{sec-intro}
Regarding the importance of massive connectivity in sixth-generation (6G) wireless technology and the escalating intricacy associated with system design in the context of ultra-massive multiple-input multiple-output (MIMO), there is an ever-growing demand for a more tractable approach that can enhance the efficiency of wireless communication systems \cite{wang2023extremely,urquiza2022spectral,marzetta2010noncooperative}. To this end, fluid antenna (FA) systems have recently emerged as a cutting-edge technology for future mobile communications, which can enhance diversity and multiplexing advantages by utilizing novel dynamic radiating structures \cite{wong2020fluid}. In particular, a FA includes a pixel-based structure or dielectric conductive \cite{huang2021liquid} that can switch its position (i.e., ports) in a pre-defined small space, where this unique feature can be especially exploited in mobile phones due to the physical limitations of antenna deployment. Moreover, compared with traditional multi-antenna systems, FA multiple access systems (FAMA) are able to eliminate the necessity for channel state information (CSI) at base stations (BSs) concerning precoding, user clustering, and power control; thereby, user equipment (UEs) are only required to perform single-user decoding \cite{wong2021fluid,wong2023slow,wong2023compact}.

Furthermore, given the practical applications of evolving technologies such as Radio Frequency Identification (RFID) systems and the Internet of Things (IoT) in realistic 6G wireless networks, considerable emphasis has been paid to backscatter communication (BC) in the recent years \cite{stockman1948communication}. Particularly, BC is a cost-effective wireless approach that enables low-power devices to sent data by reflecting or modulating existing radio frequency (RF) signals in the propagation environment \cite{van2018ambient}. In other words, backscatter devices are designed to consume minimal energy since they do not require generating their own signals. Therefore, integrating the reflective capabilities of BC with the dynamic and reconfigurable properties of FA systems can potentially provide a synergistic. FA systems are able to adaptively modify their radiating structures based on environmental conditions or network demands; this adaptability, when coupled with BC, allows for dynamic adjustments in the reflection and modulation of RF signals. Consequently, a highly flexible communication system that can optimize signal propagation, enhance spectral efficiency, and enable cost-effective wireless connectivity is provided. 

Great efforts have recently been carried out to develop the application of FA systems in different wireless communication scenarios from various aspects, e.g., channel modeling \cite{wong2022fast,khammassi2023new}, performance analysis\cite{wong2020performance,new2023fluid,10253941,10308583,ghadi2023gaussian}, channel estimation \cite{skouroumounis2022fluid,wang2023estimation}, and implementation \cite{borda2019low,jing2022compact}. However, to the best of the author's knowledge, there have been no previous works that combine the FA system with BC. Hence, motivated by the potential advantages of FA systems and the unique features of BC for the next generation of wireless technology, we evaluate the performance of wireless BC when backscatter devices take advantage of FA systems. In particular, we consider a single-antenna source that aims to send data to a FA reader through the forward (i.e., source-to-tag) and backscatter (tag-to-reader) channels. For this scenario, (i) We derive the cumulative distribution function (CDF) of the equivalent channel at the FA reader, i.e., the CDF of the maximum of $K$ random variables (RVs) such that each is the product of forward and backscatter channels, by using the copula-based approach; (ii) We obtain the outage probability (OP) and delay outage rate (DOR) in closed-form expressions under correlated Rayleigh fading channels; (iii) We derive the asymptotic expressions of the OP and DOR in the high SNR regime; (iv) We present numerical results to evaluate the performance considered for FA-aided BC, where the results indicate that the FA reader can significantly enhance the system performance compared with the single-antenna reader, namely, lower values of OP and DOR are achieved. 

\section{System Model}\label{sec-sys}
We consider a wireless FA-aided BC as illustrated in Fig. \ref{model}, where a single-antenna source aims to send information $x$ to a reader that is equipped by a FA through the forward and backscatter channels. Thus, the instantaneous received signal power at the tag is given by
\begin{align}
	P_\mathrm{t}=P_\mathrm{s}L_\mathrm{s}g_\mathrm{f},
\end{align}
in which $P_\mathrm{s}$ is the transmitted power by the source, $L_\mathrm{s}$ includes the gains of the transmit and receive antennas and frequency-dependent propagation losses, and $g_\mathrm{f}=\left|h_\mathrm{f}\right|^2$ defines the fading channel gain between the source and the tag, where $h_\mathrm{f}$ is the corresponding forward fading channel coefficient. On the reader side, we assume that the FA can freely move along $K$ pre-set positions\footnote{In this paper, the switching delay is assumed to be negligible, which is a reasonable assumption for the pixel-based FA \cite{song2013efficient}.} (i.e., ports), which are equally distributed on a linear space of length $W\lambda$ where $\lambda$ is the wavelength of propagation. Additionally, we suppose that the FA consists of only one radio frequency (RF) chain, and thus, only one port can be activated for communication. Under such assumptions, the received signal at the $k$-th port of the reader can be defined as
\begin{align}
	y_k=h_\mathrm{f}h_{\mathrm{b},k}x+z_k,
\end{align}
where $h_{\mathrm{b},k}$ denotes the backscatter channel coefficient between the tag and $k$-th port of the FA reader with the respective fading channel gain $g_{\mathrm{b},k}=\left|h_{\mathrm{b},k}\right|^2$  and $z_k$ is the independent identically distributed (i.i.d.) additive white Gaussian noise (AWGN) with zero mean and variance $\sigma^2$ at each port. Without loss of generality, we assume that the fading coefficients are normalized, i.e., $\mathbb{E}\left[g_\mathrm{f}\right]=\mathbb{E}\left[g_\mathrm{f}\right]=1$, where $\mathbb{E}\left[\cdot\right]$ denotes the expectation operator.

Furthermore, we assume that the FA is able to always switch to the best port with the strongest signal for communication, i.e.,
\begin{align}
	g_\mathrm{FA}=\max\left\{g_{\mathrm{p},k}\dots, g_{\mathrm{p},k}\right\}, \label{eq-g-fa}
\end{align}
in which $g_{\mathrm{p},k}=g_\mathrm{f}g_{\mathrm{b},k}$ denotes the product channel gain of the forward and backscatter links. It is worth noting that $g_{\mathrm{p},k}$ for $k\in\left\{1,\dots,K\right\}$ are spatially correlated since they can be arbitrarily close to each other so that such a spatial correlation between FA ports can be characterized by Jake's model as \cite{stuber2001principles}
\begin{align}
	\mu_{k}=\omega\mathcal{J}_0\left(\frac{2\pi\left(k-1\right)}{K-1}W\right),\label{eq-jake}
\end{align}
where $\mu_{k}$ denotes the correlation parameter that can control the dependency between $g_{\mathrm{p},k}$, $\omega$ is the large-scale fading effect, and $\mathcal{J}_0\left(.\right)$ represents the zero-order Bessel function of the first kind.
By doing so, the received signal-to-noise ratio (SNR) at the reader can be defined as
\begin{align}
	\gamma=\frac{P_\mathrm{t}g_\mathrm{FA}}{\sigma^2}=\bar{\gamma}g_\mathrm{FA}, \label{eq-snr}
\end{align}
in which $\bar{\gamma}=\frac{P_\mathrm{t}}{\sigma^2}$ is the average SNR. 
\begin{figure}[!t]
	\centering
	\includegraphics[width=1\columnwidth]{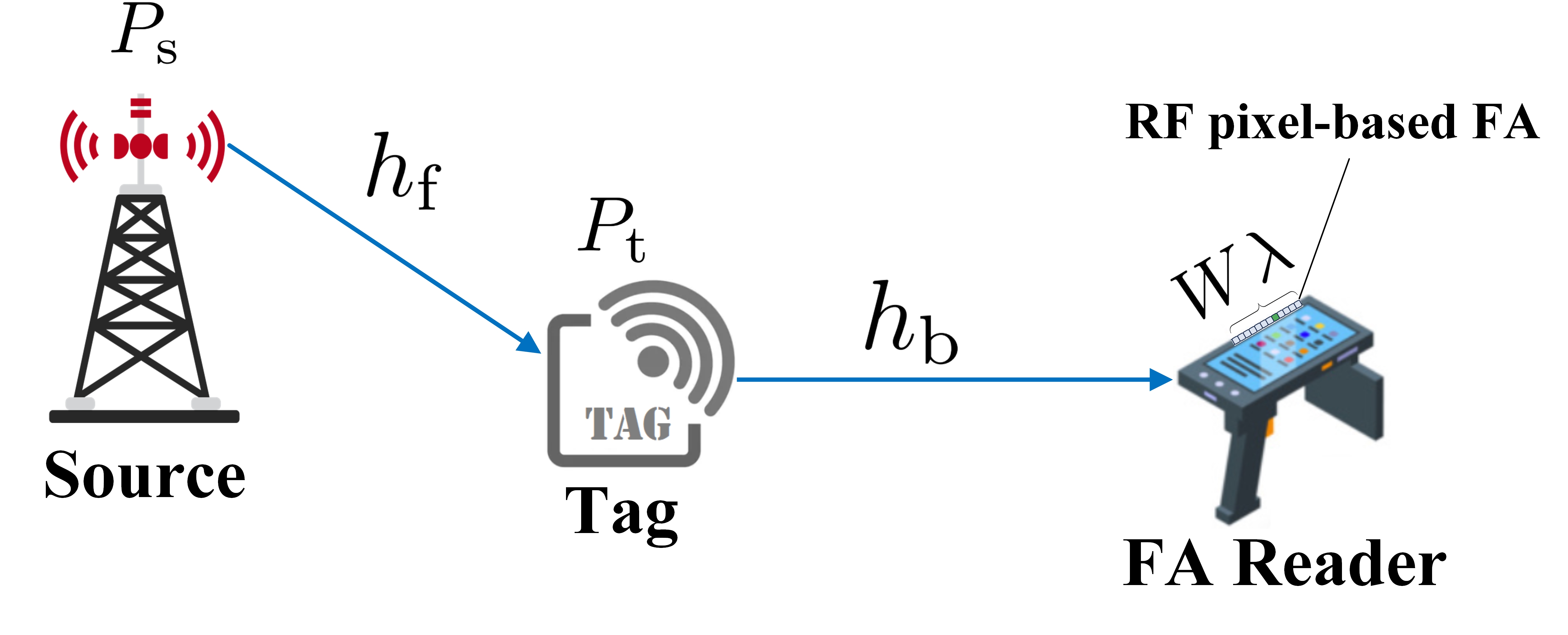}
	\caption{The system model represents FA-aided BC.}\label{model}
\end{figure}
\section{Performance Analysis}
Here, we first derive the CDF of the equivalent fading channel gain at the FA reader, and then the closed-form expressions of the OP and DOR are obtained. Moreover, we derive the asymptotic expressions of the OP and DOR in the high SNR regime.
\subsection{Statistical Characterization} 
From \eqref{eq-g-fa}, we can see that the CDF of the equivalent fading channel gain at the reader is defined as the CDF of the maximum of $K$ correlated RVs that each includes the product of two independent RVs. Assuming that all fading channels undergo Rayleigh distribution, the CDF of $g_\mathrm{FA}$ is derived as the following proposition. 
\begin{Proposition}\label{pro-cdf}
	The CDF of $g_\mathrm{FA}=\max\left\{g_{\mathrm{p},k}\dots, g_{\mathrm{p},k}\right\}$ for the considered FA-aided BC is given by
	\begin{align}
		F_{g_\mathrm{FA}}\left(r\right)=\left[\sum_{k=1}^K\left[\left(1-2\sqrt{r}\mathcal{K}_1\left(2\sqrt{r}\right)\right)^{-\theta}-1\right]+1\right]^{\frac{-1}{\theta}}\hspace{-2mm}, \label{eq-cdf}
	\end{align}
	in which $\mathcal{K}_1\left(\cdot\right)$ denotes the first-order modified Bessel function of the second kind and $\theta\in\left(0,\infty\right)$ is the dependence parameter so that $\theta\rightarrow 0$ represents the independent case.  
\end{Proposition}
\begin{proof}
	By using the CDF definition, $F_{g_\mathrm{FA}}\left(r\right)$ can be mathematically expressed as
	\begin{align}
		F_{g_\mathrm{FA}}\left(r\right)&=\Pr\left(\max\left\{g_{\mathrm{p},k}\dots, g_{\mathrm{p},k}\right\}\leq r\right)\\
		&=F_{g_{\mathrm{p},k}\dots, g_{\mathrm{p},k}}\left(r,\dots,r\right)\\
		&\overset{(a)}=C\left(	F_{g_{\mathrm{p},1}}\left(r\right),\dots,	F_{g_{\mathrm{p},K}}\left(r\right)\right), \label{eq-p1}
	\end{align}
	in which $(a)$ is obtained by using Sklar's theorem \cite[Thm. 2.10.9]{nelsen2006introduction} and $C\left(\cdot\right):\left[0,1\right]^d\rightarrow\left[0,1\right]$ denotes the copula function that is a joint CDF of $d$ random vectors on the unit cube $[0,1]^d$ with uniform marginal distributions, i.e.,
	\begin{align}
		C\left(u_1,\dots,u_d;\theta\right)=\Pr\left(U_1\leq u_1,\dots, U_d\leq u_d\right),
	\end{align}
	where $u_l=F_{s_{l}}\left(s_l\right)$ and $s_l$ is arbitrary RV for $l\in\left\{1,\dots,d\right\}$ and  $\theta$ represents the dependence parameter which can measure the linear/non-linear correlation between arbitrary correlated RVs. Hence, from \eqref{eq-p1}, we now need to find the CDF of the product channel, i.e., $F_{g_{\mathrm{p},K}}\left(r\right)$. To do so,
	assuming the Rayleigh fading channels, we can derive the CDF of $g_{\mathrm{p},k}=g_\mathrm{f}g_{\mathrm{b},k}$ as follows
	\begin{align}
		F_{g_{\mathrm{p},k}}\left(r\right)&=\Pr\left(g_\mathrm{f}g_{\mathrm{b},k}\leq r\right)\\
		&=\int_0^\infty f_{g_\mathrm{f}}\left(g_\mathrm{f}\right)F_{G_{\mathrm{b},k}}\left(\frac{r}{g_\mathrm{f}}\right)\mathrm{d}g_\mathrm{f}\\
		&=1-\int_0^\infty \mathrm{e}^{-\left(g_\mathrm{f}+\frac{r}{g_{\mathrm{f}}}\right)} \mathrm{d}g_\mathrm{f}\label{eq-p2}\\ 
		&\overset{(b)}{=}1-2\sqrt{r}\mathcal{K}_1\left(2\sqrt{r}\right), \label{eq-p3}
	\end{align}
	in which $(b)$ is obtained by solving the integral in \eqref{eq-p2} with the help of \cite[3.471.9]{gradshteyn2007table}. By inserting \eqref{eq-p3} into \eqref{eq-p1}, the CDF of $g_\mathrm{FA}$ is obtained for any arbitrary choice of copulas. However, in order to analyze the performance of the considered system model, it is required to select a copula that can describe the spatial correlation between FA ports. For this purpose, we exploit the Clayton copula because it can accurately describe the tail dependence between correlated RVs. It should be noted that an outage mainly occurs in deep fading conditions, where knowing the behavior of the tail dependence of fading coefficients is necessary; thereby, this choice is justified. Therefore, by substituting the Clayton copula definition from \cite[Def. 3]{10253941} into \eqref{eq-p1}, \eqref{eq-cdf} is obtained, and the proof is completed. 
\end{proof}
It is worth noting that since, in the copula definition, the non-linear transformations are applied to the considered RVs, the linear correlation cannot be maintained anymore. In other words, the dependence parameter $\theta$ does not necessarily represent the linear correlation between the correlated RVs. Therefore, rank correlation coefficients should be considered for the copula-based analysis since they are preserved under any monotonic transformation; consequently, they are able to describe the structure of dependency beyond linear correlation. To tackle this issue, we use Spearman's $\rho$ correlation coefficient that is identical to Pearson's product moment correlation coefficient for a pair of continuous RVs, i.e., $\rho=\mu_k$ \cite[Sec. 5.1.2]{nelsen2006introduction}. Therefore, $\rho$ between two arbitrary correlated RVs is mathematically defined as
\begin{align}
	\rho=12\iint_{\left[0,1\right]^2}u_1u_2\mathrm{d}C\left(u_1,u_2\right)-3. \label{eq-rho}
\end{align}
By plugging the Clayton copula into \eqref{eq-rho} and computing the integral, the Spearman's $\rho$ for the Clayton copula can be approximated as \cite{10364840}
\begin{align}
	\rho\approx\frac{3\theta}{2\left(\theta+2\right)}.
\end{align}
Then, considering the fact that the linear correlation coefficient is identical to Spearman's $\rho$, $\theta$ can be expressed in terms of the correlation parameter of Jake's model as follow
\begin{align}
	\theta\approx \frac{4\mu_k}{3-2\mu_k}. \label{eq-theta}
\end{align}
Thus, by substituting \eqref{eq-theta} into \eqref{eq-cdf}, the CDF of $g_\mathrm{FA}$ in Proposition \ref{pro-cdf} can be defined in terms of Jake's model.
\subsection{OP Analysis}
OP is the key performance metric in wireless communications that is defined as the probability of the instantaneous SNR $\gamma$ below the required SNR threshold $\gamma_\mathrm{th}$, i.e., 
$P_\mathrm{o}=\Pr\left(\gamma\leq\gamma_\mathrm{th}\right)$. Therefore, the OP for the considered system model is derived as the following proposition. 
\begin{Proposition}\label{pro-out}
	The OP for the considered FA-aided BC under correlated Rayleigh fading channels is given by
	\begin{align}
		P_\mathrm{o}\hspace{-1mm}=\hspace{-1mm}\left[\sum_{k=1}^K\left[\left(1-2\sqrt{\frac{\gamma_\mathrm{th}}{\bar{\gamma}}}\mathcal{K}_1\left(2\sqrt{\frac{\gamma_\mathrm{th}}{\bar{\gamma}}}\right)\right)^{ \frac{4\mu_k}{2\mu_k-3}}\hspace{-2mm}-1\right]\hspace{-1mm}+1\right]^{\frac{2\mu_k-3}{4\mu_k}}\hspace{-2mm}. \label{eq-out}
	\end{align}
	in which $\gamma_\mathrm{th}$ is the SNR threshold and $\mu_{k}$ is defined in \eqref{eq-jake}.
\end{Proposition}
\begin{proof}
	By inserting the SNR of the considered FA-aided BC from \eqref{eq-snr} into the OP definition, we have
	\begin{align}
		P_\mathrm{o}=\Pr\left(g_\mathrm{FA}\leq\frac{\gamma_\mathrm{th}}{\bar{\gamma}}\right)=F_{g_\mathrm{FA}}\left(\frac{\gamma_\mathrm{th}}{\bar{\gamma}}\right).\label{eq-prof}
	\end{align}
	Now, by applying the CDF of $g_\mathrm{FA}$ from Proposition \ref{pro-cdf} into  \eqref{eq-prof}, the proof is accomplished. 
\end{proof}
\subsection{DOR Analysis}
DOR is a momentous metric in wireless networks to evaluate the performance of ultra-reliable and low-latency communications
(URLLC) which is denoted as the probability that the transmission delay for a certain amount of data $R$
in a wireless channel with a bandwidth $B$ exceeds a certain predefined threshold $T_\mathrm{th}$, i.e., $\Pr\left(T_\mathrm{dt}>T_\mathrm{th}\right)$, in which $
T_\mathrm{dt}=\frac{R}{B\log_2\left(1+\gamma\right)}$ defines the delivery time \cite{yang2019ultra}. Therefore, the DOR for the considered system model can be achieved as the following proposition. 
\begin{Proposition}\label{pro-dor}
	The DOR for the considered FA-aided BC under correlated Rayleigh fading channels is given by
	\begin{align}\notag
		&P_\mathrm{dor}\\
		&=\hspace{-1mm}\left[\sum_{k=1}^K\left[\left(1-2\sqrt{\frac{\hat{T}_\mathrm{th}}{\bar{\gamma}}}\mathcal{K}_1\left(2\sqrt{\frac{\hat{T}_\mathrm{th}}{\bar{\gamma}}}\right)\right)^{ \frac{4\mu_k}{2\mu_k-3}}\hspace{-4mm}-1\right]\hspace{-1mm}+1\right]^{\frac{2\mu_k-3}{4\mu_k}}\hspace{-5mm}, \label{eq-dor}
	\end{align}
	in which $\hat{T}_\mathrm{th}=\mathrm{e}^{\frac{R\ln 2}{BT_\mathrm{th}}}$ is defined in \eqref{eq-jake}
\end{Proposition}
\begin{proof}
	By inserting the delivery time into the DOR definition, we have
	\begin{align}
		P_\mathrm{dor}&=\Pr\left(\frac{R}{B\log_2\left(1+\bar{\gamma}g_\mathrm{FA}\right)}>T_\mathrm{th}\right)\\
		&=\Pr\left(g_\mathrm{FA}\leq{\frac{\mathrm{e}^{\frac{R\ln 2}{BT_\mathrm{th}}}}{\bar{\gamma}}}\right)=F_{g_\mathrm{FA}}\left(\frac{\hat{T}_\mathrm{th}}{\bar{\gamma}}\right).
	\end{align}
	Now, by inserting $\hat{T}_\mathrm{th}=\mathrm{e}^{\frac{R\ln 2}{BT_\mathrm{th}}}$ into the CDF of $g_\mathrm{FA}$ from Proposition \ref{pro-cdf}, the proof is completed. 
\end{proof}
\subsection{Asymptotic Analysis}
Although the derived OP and DOR in Propositions \ref{pro-out} and \ref{pro-dor} are in simple closed-form expressions, we are interested in the asymptotic behavior of the obtained metrics at the high SNR regime (i.e., $\gamma\rightarrow\infty$) to gain more insights into the system performance. To do so, we exploit the series expansion of the Bessel function $\mathcal{K}_1\left(r\right)$ when $r\rightarrow 0$ as follow
\begin{align}
	\mathcal{K}_1\left(r\right)\approx \frac{1}{r}+\frac{r}{4}\left(2\zeta-1\right)+\frac{r}{2}\log\left(\frac{r}{2}\right), \label{eq-bessel}
\end{align}
where $\zeta$ is the Euler-Mascheroni constant \cite{weisstein2002euler}. Hence, the asymptotic expressions of the OP and DOR for the considered system model can be obtained in the following corollary. 
\begin{corollary}
	The asymptotic expressions of the OP and DOR for the considered FA-aided BC at the high SNR regime, i.e., $\bar{\gamma}\rightarrow\infty$ are respectively given by \eqref{eq-asy-out} and \eqref{eq-asy-dor}.
	\begin{figure*}[!b]
		\normalsize
		\setcounter{equation}{23}
		\begin{align}
			P_\mathrm{o}^\infty\approx\left[\sum_{k=1}^K\left[\left(\frac{\gamma_\mathrm{th}}{\bar{\gamma}}\left[1-2\zeta-2\log\left(\sqrt{\frac{\gamma_\mathrm{th}}{\bar{\gamma}}}\right)\right]\right)^{ \frac{4\mu_k}{2\mu_k-3}}-1\right]+1\right]^{\frac{2\mu_k-3}{4\mu_k}}. \label{eq-asy-out}
		\end{align}
		\hrulefill
		\setcounter{equation}{24}
		\begin{align}
			P_\mathrm{dor}^\infty\approx\left[\sum_{k=1}^K\left[\left(\frac{\hat{T}_\mathrm{th}}{\bar{\gamma}}\left[1-2\zeta-2\log\left(\sqrt{\frac{\hat{T}_\mathrm{th}}{\bar{\gamma}}}\right)\right]\right)^{ \frac{4\mu_k}{2\mu_k-3}}-1\right]+1\right]^{\frac{2\mu_k-3}{4\mu_k}}. \label{eq-asy-dor}
		\end{align}
		\hrulefill
	\end{figure*}
\end{corollary}
\begin{proof}
	For the high SNR regime (i.e., $\bar{\gamma}\rightarrow\infty$), we have $\frac{\omega}{\bar{\gamma}}\rightarrow 0$, where $\omega\in\left\{\gamma_\mathrm{th},\hat{T}_\mathrm{th}\right\}$. Hence, by utilizing \eqref{eq-bessel},  $\frac{2\omega}{\bar{\gamma}}\mathcal{K}_1\left(\frac{2\omega}{\bar{\gamma}}\right)$ can be approximated as 
	\begin{align}
		\frac{2\omega}{\bar{\gamma}}\mathcal{K}_1\left(\frac{2\omega}{\bar{\gamma}}\right)\approx1+ \frac{\omega}{\bar{\gamma}}\left[2\zeta-1+2\log\left(\sqrt{\frac{\omega}{\bar{\gamma}}}\right)\right].\label{eq-prof2}
	\end{align}
	Next, by substituting  \eqref{eq-prof2} into \eqref{eq-out} and \eqref{eq-dor}, the proof is accomplished. 
\end{proof}
\section{Numerical Results}
In this section, we present numerical results to evaluate the considered system performance in terms of the OP and DOR, which are double-checked by the Monte-Carlo simulation method. To this end, we set the parameters as $\gamma_\mathrm{th}=0$dB, $R=5$Kbits, $B=2$GHz, $T_\mathrm{dt}=3$ms, $\bar{\gamma}=20$dB, $W=\left\{0.5,1,2,4,6\right\}$, and $N=\left\{2,4,6,8,10\right\}$.

Figs. \ref{subfig-out-g-w} and \ref{subfig-dor-g-w} respectively illustrate the behavior of OP and DOR in terms of the average SNR $\bar{\gamma}$ for given values of FA size $W$ under correlated Rayleigh fading channels. As expected, the OP and DOR decrease as $\bar{\gamma}$ increases, which is reasonable since the channel condition improves. Moreover, it can be seen that by increasing the FA size $W$ for a fixed number of ports $K$, the performance of OP and DOR improves. The reason for this behavior is that increasing the spatial separation between the FA ports by increasing $W$ for a fixed $K$ can reduce the spatial correlation between FA ports. Additionally, we can clearly observe that such an improvement is more noticeable when $K$ is large. The performance of OP and DOR in terms of $\bar{\gamma}$ for given values of $K$ under correlated Rayleigh fading channels is presented in Figs. \ref{subfig-out-g-k} and \ref{subfig-dor-g-k}, respectively. We can see that as the number of FA ports $K$ grows, lower values of the OP and DOR are provided. The main reason for such a behavior is that although increasing $K$ for a fixed value of $W$ raises the spatial correlation between FA ports, it can potentially improve the channel capacity, diversity gain, and spatial multiplexing at the same time. Hence, this can help mitigate fading and improve the overall link quality.
Furthermore, as we can see in both Figs. \ref{fig-snr1} and \ref{fig-snr2}, considering a FA reader instead of a single-antenna reader can significantly enhance the performance of BC in terms of OP and DOR. In order to evaluate how the FA reader affects the DOR performance in terms of transmitted data $R$ over BC, we present Fig. \ref{fig-dor-r} for selected values of $W$ and $K$. First, we can observe that as $W$ and $K$ increase simultaneously, the spatial correlation between FA ports becomes balanced; thereby, lower values of the OP and DOR are reached. Furthermore, as expected, it can be seen that as $R$ increases, the DOR performance becomes worse, such that
transmitting a high amount of data (e.g., $R=3$Kbits) with low delay is almost impractical when a single-antenna reader or a FA reader with small $W$ and $K$ are considered. However, thanks to the FA reader, a large amount of information with a small delay can be sent when the $W$ and $K$ are large enough.
\begin{figure*}[t!]
	\centering
	\subfigure[OP]{%
		\includegraphics[width=0.47\textwidth]{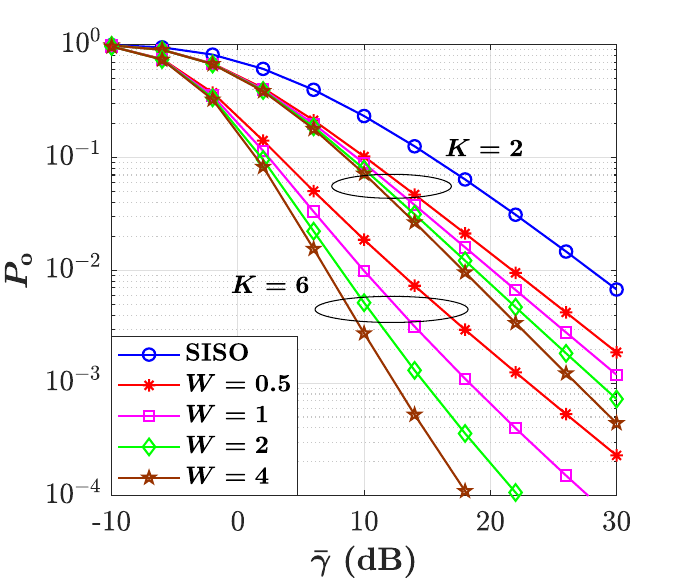}\label{subfig-out-g-w}%
	}
	\subfigure[DOR]{%
		\includegraphics[width=0.47\textwidth]{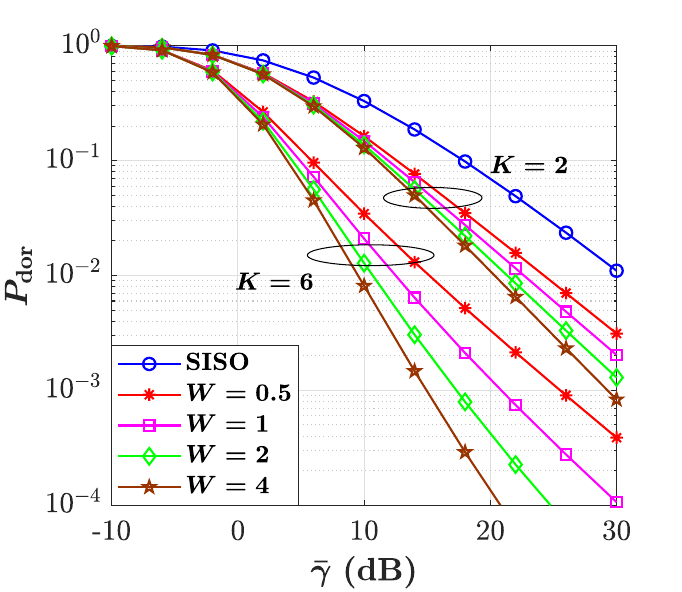}\label{subfig-dor-g-w}%
	}\hspace{0cm}
	\caption{Performance of (a) OP and (b) DOR versus average SNR $\bar{\gamma}$ for selected values of FA size $W$.}\label{fig-snr1}\vspace{0cm}
\end{figure*}
\begin{figure*}[t!]
	\centering
	\subfigure[OP]{%
		\includegraphics[width=0.47\textwidth]{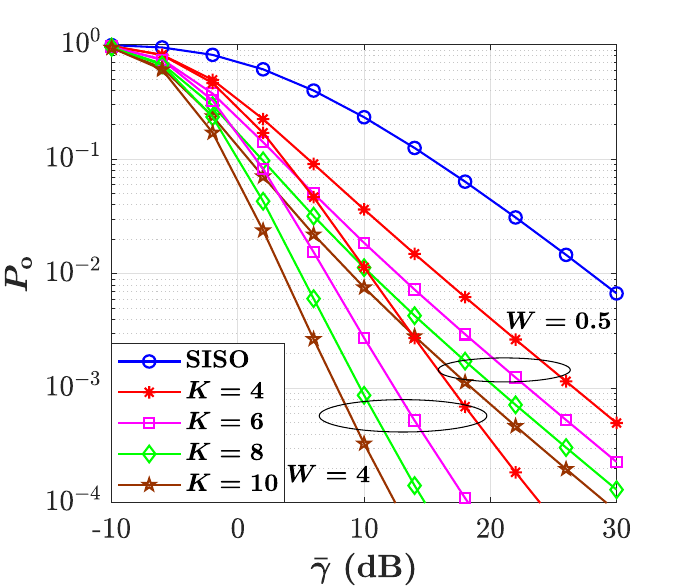}\label{subfig-out-g-k}%
	}
	\subfigure[DOR]{%
		\includegraphics[width=0.47\textwidth]{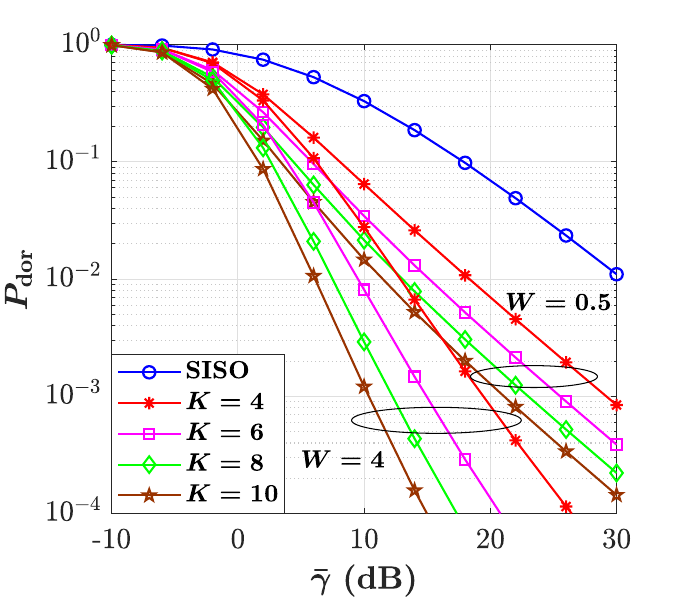}\label{subfig-dor-g-k}%
	}\hspace{0cm}
	\caption{Performance of (a) OP and (b) DOR versus average SNR $\bar{\gamma}$ for selected values of FA ports $K$.}\label{fig-snr2}\vspace{0cm}
\end{figure*}
\begin{figure}[!t]
	\centering
	\includegraphics[width=1\columnwidth]{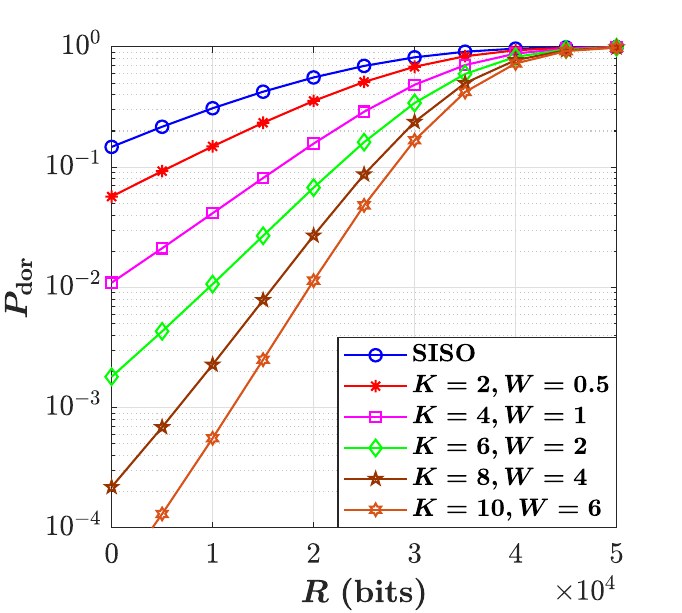}
	\caption{DOR versus amount of data $R$ for selected of values $W$ and  $K$.}\label{fig-dor-r}
\end{figure}
\section{Conclusion}
In this paper, we investigated the performance of BC in the presence of FA system. Particularly, we assumed that a single-antenna source aims to send information to a reader via wireless forward and backscatter channels. We also supposed that the reader includes a FA, where only one port can be activated for communication. Under such assumptions, we first derive the CDF of the equivalent channel (i.e., the maximum of $K$ correlated RVs such that each is the product of forward and backscatter channels) for the reader by exploiting the copula technique. Then, we derived the closed-form expressions of the OP and DOR under correlated Rayleigh fading channels. Furthermore, we obtained the asymptotic expressions of the OP and DOR in the high SNR regime. Eventually, our analytical results revealed that the FA reader can provide a remarkable performance in terms of the OP and DOR compared with the single-antenna reader over BC.

\bibliographystyle{IEEEtran}
\bibliography{sample.bib}

\begin{thebibliography}{10}
\providecommand{\url}[1]{#1}
\csname url@samestyle\endcsname
\providecommand{\newblock}{\relax}
\providecommand{\bibinfo}[2]{#2}
\providecommand{\BIBentrySTDinterwordspacing}{\spaceskip=0pt\relax}
\providecommand{\BIBentryALTinterwordstretchfactor}{4}
\providecommand{\BIBentryALTinterwordspacing}{\spaceskip=\fontdimen2\font plus
\BIBentryALTinterwordstretchfactor\fontdimen3\font minus
  \fontdimen4\font\relax}
\providecommand{\BIBforeignlanguage}[2]{{%
\expandafter\ifx\csname l@#1\endcsname\relax
\typeout{** WARNING: IEEEtran.bst: No hyphenation pattern has been}%
\typeout{** loaded for the language `#1'. Using the pattern for}%
\typeout{** the default language instead.}%
\else
\language=\csname l@#1\endcsname
\fi
#2}}
\providecommand{\BIBdecl}{\relax}
\BIBdecl

\bibitem{wang2023extremely}
Z.~Wang, J.~Zhang, H.~Du, E.~Wei, B.~Ai, D.~Niyato, and M.~Debbah, ``{Extremely
  large-scale MIMO: Fundamentals, challenges, solutions, and future
  directions},'' \emph{IEEE Wirel. Commun.}, 2023.

\bibitem{urquiza2022spectral}
D.~A. Urquiza~Villalonga, H.~OdetAlla, M.~J. Fern{\'a}ndez-Getino~Garc{\'\i}a,
  and A.~Flizikowski, ``{Spectral efficiency of precoded 5G-NR in single and
  multi-user scenarios under imperfect channel knowledge: A comprehensive guide
  for implementation},'' \emph{Electronics}, vol.~11, no.~24, p. 4237, 2022.

\bibitem{marzetta2010noncooperative}
T.~L. Marzetta, ``{Noncooperative cellular wireless with unlimited numbers of
  base station antennas},'' \emph{IEEE Trans. Wirel. Commun.}, vol.~9, no.~11,
  pp. 3590--3600, 2010.

\bibitem{wong2020fluid}
K.-K. Wong, A.~Shojaeifard, K.-F. Tong, and Y.~Zhang, ``{Fluid antenna
  systems},'' \emph{IEEE Trans. Wirel. Commun.}, vol.~20, no.~3, pp.
  1950--1962, 2020.

\bibitem{huang2021liquid}
Y.~Huang, L.~Xing, C.~Song, S.~Wang, and F.~Elhouni, ``{Liquid antennas: Past,
  present and future},'' \emph{IEEE Open J. Antennas Propag.}, vol.~2, pp.
  473--487, 2021.

\bibitem{wong2021fluid}
K.-K. Wong and K.-F. Tong, ``{Fluid antenna multiple access},'' \emph{IEEE
  Trans. Wirel. Commun.}, vol.~21, no.~7, pp. 4801--4815, 2021.

\bibitem{wong2023slow}
K.-K. Wong, D.~Morales-Jimenez, K.-F. Tong, and C.-B. Chae, ``{Slow fluid
  antenna multiple access},'' \emph{IEEE Trans. Commun.}, 2023.

\bibitem{wong2023compact}
K.-K. Wong, C.-B. Chae, and K.-F. Tong, ``{Compact ultra massive antenna array:
  A simple open-loop massive connectivity scheme},'' \emph{IEEE Trans. Wirel.
  Commun.}, 2023.

\bibitem{stockman1948communication}
H.~Stockman, ``{Communication by means of reflected power},'' \emph{Proc. IRE},
  vol.~36, no.~10, pp. 1196--1204, 1948.

\bibitem{van2018ambient}
N.~Van~Huynh, D.~T. Hoang, X.~Lu, D.~Niyato, P.~Wang, and D.~I. Kim, ``{Ambient
  backscatter communications: A contemporary survey},'' \emph{IEEE Commun.
  Surv. Tutor.}, vol.~20, no.~4, pp. 2889--2922, 2018.

\bibitem{wong2022fast}
K.-K. Wong, K.-F. Tong, Y.~Chen, and Y.~Zhang, ``{Fast fluid antenna multiple
  access enabling massive connectivity},'' \emph{IEEE Commun. Lett.}, vol.~27,
  no.~2, pp. 711--715, 2022.

\bibitem{khammassi2023new}
M.~Khammassi, A.~Kammoun, and M.-S. Alouini, ``{A new analytical approximation
  of the fluid antenna system channel},'' \emph{IEEE Trans. Wirel. Commun.},
  2023.

\bibitem{wong2020performance}
K.~K. Wong, A.~Shojaeifard, K.-F. Tong, and Y.~Zhang, ``{Performance limits of
  fluid antenna systems},'' \emph{IEEE Commun. Lett.}, vol.~24, no.~11, pp.
  2469--2472, 2020.

\bibitem{new2023fluid}
W.~K. New, K.-K. Wong, H.~Xu, K.-F. Tong, and C.-B. Chae, ``{Fluid antenna
  system: New insights on outage probability and diversity gain},'' \emph{IEEE
  Trans. Wirel. Commun.}, 2023.

\bibitem{10253941}
F.~Rostami~Ghadi, K.-K. Wong, F.~J. López-Martínez, and K.-F. Tong,
  ``{Copula-Based Performance Analysis for Fluid Antenna Systems Under
  Arbitrary Fading Channels},'' \emph{IEEE Commun. Lett.}, vol.~27, no.~11, pp.
  3068--3072, 2023.

\bibitem{10308583}
J.~D. Vega-Sánchez, A.~E. López-Ramírez, L.~Urquiza-Aguiar, and D.~P.~M.
  Osorio, ``{Novel Expressions for the Outage Probability and Diversity Gains
  in Fluid Antenna System},'' \emph{IEEE Wirel. Commun. Lett.}, pp. 1--1, 2023.

\bibitem{ghadi2023gaussian}
F.~R. Ghadi, K.-K. Wong, F.~J. Lopez-Martinez, C.-B. Chae, K.-F. Tong, and
  Y.~Zhang, ``{A Gaussian Copula Approach to the Performance Analysis of Fluid
  Antenna Systems},'' \emph{arXiv preprint arXiv:2309.07506}, 2023.

\bibitem{skouroumounis2022fluid}
C.~Skouroumounis and I.~Krikidis, ``{Fluid antenna with linear MMSE channel
  estimation for large-scale cellular networks},'' \emph{IEEE Trans. Commun.},
  vol.~71, no.~2, pp. 1112--1125, 2022.

\bibitem{wang2023estimation}
R.~Wang, Y.~Chen, Y.~Hou, K.-K. Wong, and X.~Tao, ``{Estimation of channel
  parameters for port selection in millimeter-wave fluid antenna systems},'' in
  \emph{2023 IEEE/CIC Int. Conf. Commun. in China (ICCC Workshops)}.\hskip 1em
  plus 0.5em minus 0.4em\relax IEEE, 2023, pp. 1--6.

\bibitem{borda2019low}
C.~Borda-Fortuny, L.~Cai, K.~F. Tong, and K.-K. Wong, ``{Low-cost 3D-printed
  coupling-fed frequency agile fluidic monopole antenna system},'' \emph{IEEE
  Access}, vol.~7, pp. 95\,058--95\,064, 2019.

\bibitem{jing2022compact}
L.~Jing, M.~Li, and R.~Murch, ``{Compact Pattern Reconfigurable Pixel Antenna
  With Diagonal Pixel Connections},'' \emph{IEEE Trans. Antennas Propag.},
  vol.~70, no.~10, pp. 8951--8961, 2022.

\bibitem{song2013efficient}
S.~Song and R.~D. Murch, ``{An efficient approach for optimizing frequency
  reconfigurable pixel antennas using genetic algorithms},'' \emph{IEEE Trans.
  Antennas Propag.}, vol.~62, no.~2, pp. 609--620, 2013.

\bibitem{stuber2001principles}
G.~L. St{\"u}ber, \emph{{Principles of mobile communication}}.\hskip 1em plus
  0.5em minus 0.4em\relax Springer, 2001, vol.~2.

\bibitem{nelsen2006introduction}
R.~B. Nelsen, \emph{{An introduction to copulas}}.\hskip 1em plus 0.5em minus
  0.4em\relax Springer, 2006.

\bibitem{gradshteyn2007table}
I.~S. Gradshteyn and I.~M. Ryzhik, \emph{{Table of integrals, series, and
  products}}.\hskip 1em plus 0.5em minus 0.4em\relax Academic, 7th ed., 2007.

\bibitem{10364840}
F.~R. Ghadi, K.-K. Wong, F.~J. López-Martínez, C.-B. Chae, K.-F. Tong, and
  Y.~Zhang, ``{Fluid Antenna-Assisted Dirty Multiple Access Channels over
  Composite Fading},'' \emph{IEEE Commun. Lett.}, pp. 1--1, 2023.

\bibitem{yang2019ultra}
H.-C. Yang, S.~Choi, and M.-S. Alouini, ``{Ultra-reliable low-latency
  transmission of small data over fading channels: A data-oriented analysis},''
  \emph{IEEE Commun. Lett.}, vol.~24, no.~3, pp. 515--519, 2019.

\bibitem{weisstein2002euler}
E.~W. Weisstein, ``Euler-mascheroni constant,'' \emph{https://mathworld.
  wolfram. com/}, 2002.

\end{thebibliography}
\end{document}